\documentclass[11pt]{article}
\usepackage[letterpaper,margin=0.75in]{geometry}
\usepackage[T1]{fontenc}
\usepackage[utf8]{inputenc}
\usepackage{newpxtext}
\usepackage{microtype}
\usepackage{authblk}
\usepackage{amsthm}
\usepackage{amsmath,amssymb,mathtools}
\usepackage{newpxmath}

\usepackage{booktabs}
\usepackage{graphicx}
\graphicspath{{figures/}}
\usepackage{float}         
\usepackage{algorithm}
\usepackage[noend]{algpseudocode}
\floatname{algorithm}{Algorithm}

\usepackage{xcolor}
\definecolor{CiteGreen}{RGB}{0,120,60}
\definecolor{LinkRed}{RGB}{160,20,20}
\definecolor{UrlBlue}{RGB}{20,60,170}

\usepackage[
  colorlinks=true,
  citecolor=CiteGreen,
  linkcolor=LinkRed,
  urlcolor=UrlBlue,
  bookmarksnumbered=true
]{hyperref}
\usepackage{babel}

\usepackage{cleveref}
\theoremstyle{plain}
\newtheorem{theorem}{Theorem}
\newtheorem*{theorem*}{Theorem}
\newtheorem{lemma}{Lemma}
\newtheorem{proposition}{Proposition}
\newtheorem{corollary}{Corollary}

\theoremstyle{definition}
\newtheorem{definition}{Definition}
\newtheorem{example}{Example}

\theoremstyle{plain}

\newcommand{\restatedname}{}
\newtheorem*{restated}{\restatedname}

\makeatletter
\renewenvironment{proof}[1][\proofname]{\par
  \pushQED{\qed}%
  \normalfont \topsep6\p@\@plus6\p@\relax
  \trivlist
  \item[\hskip\labelsep\bfseries #1\@addpunct{}]\ignorespaces
}{%
  \popQED\endtrivlist\@endpefalse
}
\makeatother

\newcommand{\agents}{N}                       
\newcommand{\items}{M}                        
\newcommand{\alloc}{\mathcal{A}}              
\newcommand{\allocB}{\mathcal{B}}
\newcommand{\bundle}[1]{A_{#1}}
\newcommand{\subsidy}{p}                      
\newcommand{\optsubsidy}{p^{*}}               
\newcommand{\envygraph}[1]{G_{#1}}            
\newcommand{\edgew}[1]{w_{#1}}                
\newcommand{\pathw}[1]{\ell_{#1}}             
\newcommand{\SW}{\mathrm{SW}}                 
        
\newcommand{\margplus}[3]{\Delta^{+}_{#1}(#2,#3)}

\newcommand{\cost}{c}

\newcommand{\R}{\mathbb{R}}

\newcommand{\set}[1]{\left\{#1\right\}}
\newcommand{\abs}[1]{\left\lvert#1\right\rvert}

\usepackage{fancyhdr}
\begin{document}

\makeatletter
\renewcommand\@makefnmark{\hbox{\@textsuperscript{\normalfont\@thefnmark}}}
\makeatother
\renewcommand{\thefootnote}{\fnsymbol{footnote}}

\title{\textsc{Envy-Free Chore Allocation with Unit Subsidies under Negative Dichotomous Valuations
}}

\author[1]{Mainak Sarkar\footnote{Contact: mainaks23@iitk.ac.in}}
\author[1]{Soumyarup Sadhukhan\footnote{Contact: soumyarup.sadhukhan@gmail.com}}
\affil[1]{\small Indian Institute of Technology Kanpur}

\date{}

\maketitle
\renewcommand{\thefootnote}{\arabic{footnote}}
\setcounter{footnote}{0}

\begin{abstract}
We study the allocation of indivisible chores with subsidies under negative dichotomous valuations, where the marginal disutility of every chore is either zero or one. This is the chore analogue of the dichotomous goods model of Barman et al.~\cite{BKNS22}, for which a subsidy of at most one unit per agent is known to suffice for envy-freeness.

Our main contribution is to show that the same optimal guarantee holds for chores under no structural assumption beyond binary marginals. For every instance with negative dichotomous valuations, there exists a complete allocation $\mathcal{A}$ and a subsidy vector $p\in\{0,1\}^n$ with $\sum_{i\in N} p_i\le n-1$ such that $(\mathcal{A},p)$ is envy-free. Moreover, $\mathcal{A}$ is EF1 even before subsidies are paid. The subsidy bound is tight, and both the allocation and the subsidies can be computed in polynomial time in the value-oracle model. Unlike in the corresponding goods setting, a natural incremental extension approach does not work for chores. Our proof instead exploits key properties of an envy-free partial allocation and its associated equality graph to complete the allocation and determine which agents require subsidies.

As a further extension of the main theorem, we establish the same optimal unit-subsidy guarantee for objective mixed instances with signed binary marginals, where every item is a good for all agents or a chore for all agents; these valuations need not be additive. The key step is to prove that the goods-side extension argument of Barman et al.~\cite{BKNS22} remains valid even when the previously allocated bundles contain chores. Finally, even for negative dichotomous chore valuations, we show that the unit-subsidy guarantee cannot, in general, be maintained if Pareto optimality is also required.\end{abstract}

\medskip
\noindent\textbf{Keywords:}
Fair division, indivisible chores, envy-freeness, subsidies, negative dichotomous valuations, objective signed dichotomous valuations, EF1.

\medskip
\noindent\textbf{Mathematics Subject Classification (2020):}
91B32, 68W40, 91A12.
\section{Introduction}
\label{sec:intro}
Discrete fair division studies the allocation of indivisible resources among agents with heterogeneous preferences, and lies at the interface of mathematical economics and computer science~\cite{BCE16,End17}. A central fairness criterion is \emph{envy-freeness}~\cite{Fol67,Var74}: no agent prefers another agent's bundle to her own. For indivisible items, however, an envy-free allocation need not exist, as the standard one-item two-agent instance illustrates. This has motivated a large literature on relaxations of envy-freeness, including envy-freeness up to one good (EF1)~\cite{Bud11,LMMS04} and the stronger notion of envy-freeness up to any good (EFX)~\cite{CKMPSW19}.

A second and complementary line of work retains exact envy-freeness by allowing monetary subsidies. In addition to her allocated bundle $A_i$, each agent $i$ receives a subsidy $p_i\geq 0$ from a third party, and the outcome is envy-free if
\[v_i(A_i)+p_i \geq v_i(A_j)+p_j
    \qquad\text{for all }i,j.\]
Halpern and Shah~\cite{HS19} show that every instance admits an envy-freeable allocation and study the minimum subsidy required to achieve envy-freeness. Under the standard normalization that each item value is at most one,\footnote{Such a normalization is necessary when stating absolute bounds on monetary subsidies.} they show that, for a given envy-freeable allocation, a total subsidy of at most $(n-1)m$ may be required in the worst case, and conjecture that a total subsidy of $n-1$ always suffices when the allocation can also be chosen. Brustle et al.~\cite{BDNSV20} prove this conjecture for additive valuations, obtaining an allocation that is also EF1 and balanced, and show that for general monotone valuations a subsidy of at most $2(n-1)$ per agent suffices.

Barman et al.~\cite{BKNS22} obtain a substantially stronger guarantee for dichotomous valuations. When every marginal $v_i(S\cup\{g\})-v_i(S)$ belongs to $\{0,1\}$, they show that there exists an allocation admitting subsidies in $\{0,1\}^n$, with total subsidy at most $n-1$, and that such an outcome can be computed in polynomial time in the value-oracle model. This bound is tight: for $n$ agents and a single good valued at one by every agent, the $n-1$ agents who do not receive the good each require one unit of subsidy.

In this paper, we consider the corresponding problem for indivisible chores, where receiving an additional item weakly decreases an agent's utility. Chore-allocation problems arise naturally in applications such as shift assignment, on-call duties, teaching and refereeing loads, committee service, and maintenance tasks. Monetary compensation is also natural in such settings, where it may be interpreted as payment for undertaking an undesirable task or additional workload.

We study \emph{negative dichotomous} valuations: for every agent $i$, every $S\subseteq \items$, and every $g\in \items\setminus S$, \[ v_i(S\cup\set{g})-v_i(S)\in\set{0,-1}. \] Thus, every item is a chore for every agent, and adding a chore changes an agent's utility by either zero or minus one. This domain is the sign analogue of the dichotomous goods model of~\cite{BKNS22}.

Binary marginals capture settings in which an additional task either creates a new unit of burden or can be absorbed by an existing commitment. For example, an agent already covering an on-call window may incur no additional cost from another incident during that window, or a driver already routed through a district may absorb an additional stop at no extra cost. Since whether a chore is costly may depend on the rest of the assigned bundle, the model naturally allows non-additive valuations, including valuations that are not subadditive.

Although chore allocation is formally related to goods allocation by reversing the direction of preferences, this correspondence does not generally extend to algorithmic arguments. Several works have observed that algorithms and structural results for goods need not carry over directly to chores~\cite{aziz2017algorithms,sun2023fairness, BSV21}. The same difficulty arises in our setting. The algorithm of Barman et al.~\cite{BKNS22} for dichotomous goods is incremental: it allocates the goods one at a time while maintaining an envy-freeable partial allocation whose required subsidies lie in $\{0,1\}^n$. For chores, one might naturally try to preserve the same invariant by assigning each new chore to an agent who currently receives no subsidy. This approach, however, may not succeed. There are negative dichotomous instances in which a partial allocation admits an envy-free solution with subsidies in $\{0,1\}^n$, but after introducing one additional chore, no assignment of that chore to an agent preserves the unit-subsidy guarantee; this remains true even if the previously allocated bundles are allowed to be permuted before assigning the new chore. We give such an example in Appendix~\ref{sec:whynotbkns}. Thus, the one-item extension principle underlying the goods-side construction fails for chores, and a different argument is required.

A general subsidy bound for our setting follows from the work of Kawase et al.~\cite{kawase2025towards}. They show that, for doubly monotone valuations with marginal values bounded in absolute value by one, an EF1 allocation can be converted into an envy-free solution with a subsidy of at most $n-1$ per agent and at most $n(n-1)/2$ in total. Negative dichotomous valuations form a subclass of this domain. Thus, Kawase et al.~\cite{kawase2025towards} already imply, for our setting, a subsidy bound that is independent of the number of chores: at most $n-1$ per agent and $n(n-1)/2$ in total. The question is whether this can be strengthened to the unit-subsidy guarantee known for dichotomous goods. Our main result answers this question affirmatively.

\begin{theorem*}
For every instance with negative dichotomous valuations, there exists a complete allocation $\mathcal A$ and a subsidy vector
$p\in\{0,1\}^n$ such that $(\mathcal A,p)$ is envy-free and
\[
    \sum_{i\in N} p_i\le n-1.
\]
Moreover, $\mathcal A$ is EF1 even before subsidies are paid, and such an allocation and subsidy vector can be computed in polynomial time in the value-oracle model.
\end{theorem*}
This additional fairness guarantee is not provided by the construction of Barman et al.~\cite{BKNS22} for dichotomous goods. Thus, beyond extending the unit-subsidy guarantee to negative dichotomous chores, our construction also ensures that the underlying allocation is EF1.

The total subsidy bound is tight even for identical binary additive valuations. Consider $n$ agents and $n-1$ chores with $v_i(S)=-|S|$ for every $i$. In every complete allocation, some agent $h$ receives the empty bundle. Envy-freeness with subsidies then gives, for every $i$,
\[
    -|A_i|+p_i\ge p_h,
    \qquad\text{so}\qquad
    p_i\ge |A_i|.
\]
Summing over the agents yields
$\sum_i p_i\ge\sum_i|A_i|=n-1$, proving tightness.

Our proof follows a different route from the incremental construction of Barman et al.~\cite{BKNS22}. We start from the polynomial-time algorithm of Tao et al.~\cite{tao2025existence}, which returns an envy-free partial allocation while leaving at most $n-1$ chores unassigned. At termination, we assign the remaining chores to distinct agents in a suitable tail strongly connected component of the associated equality graph. We then identify the subsidised agents by tracing equality edges backwards from the agents who receive the remaining chores. The resulting complete allocation is EF1, and assigning one unit of subsidy to the identified agents makes it exactly envy-free.

As a further extension of the main theorem, we establish the unit-subsidy guarantee for objective mixed instances with signed binary marginals. The key step is to show that the one-good extension argument of Barman et al.~\cite{BKNS22} continues to apply when the already allocated bundles may contain chores. This yields the same optimal unit-subsidy guarantee for a non-additive mixed domain in which every item has a common sign across agents; the total bound of \(n-1\) remains tight because negative dichotomous chore instances form a subclass.

As a boundary result, we study the compatibility of unit subsidies with Pareto optimality under dichotomous chore valuations. In the binary additive subclass, a simple balancing argument gives a Pareto-optimal and EF1 allocation that can be made envy-free with subsidies in \(\{0,1\}^n\) of total at most \(n-1\); see Proposition~\ref{prop:additive-po}. In contrast, this compatibility fails for general negative dichotomous valuations: even with two agents and identical binary submodular valuations, every Pareto-optimal allocation may require a subsidy of at least \(2\) for some agent; see Proposition~\ref{prop:po-impossibility}.

Section~\ref{sec:prelim} introduces the model and the necessary preliminaries. Section~\ref{sec:main} presents the algorithm and proves the main result for negative dichotomous valuations. Finally, Section~\ref{sec:objective-mixed} deals with objective mixed valuations. Appendix~\ref{sec:whynotbkns} gives the example showing that incremental extension fails for chores, and Appendices~\ref{app:goods-extension}--\ref{app:additive-objective-po} contain the proofs and the algorithm deferred from Section~\ref{sec:objective-mixed}.

\subsection{Related work}

The literature on envy-freeness with subsidies has been developed primarily for goods. Halpern and Shah~\cite{HS19} give a characterisation of envy-freeable allocations and obtain the optimal total subsidy bound of $n-1$ for binary additive valuations. Goko et al.~\cite{goko2024fair} study matroidal valuations and obtain a truthful mechanism achieving envy-freeness and utilitarian optimality with at most one unit of subsidy per agent, under free disposal. If all goods are required to be allocated, they retain the unit-subsidy and efficiency guarantees, but without truthfulness. Barman et al.~\cite{BKNS22} show that, if truthfulness is not required, no structure beyond binary marginals is needed for the unit-subsidy guarantee. Kawase et al.~\cite{kawase2025towards} consider the substantially broader class of doubly monotone valuations and show that any EF1 allocation can be converted into an envy-free solution using at most $n-1$ units of subsidy per agent and $n(n-1)/2$ in total. Since negative dichotomous valuations form a subclass of doubly monotone valuations, their result gives the previously known $m$-independent subsidy benchmark for our domain, but does not yield the unit-subsidy guarantee established here.

Fair allocation of indivisible chores without subsidies has also received considerable attention. Bhaskar et al.~\cite{BSV21} give polynomial-time algorithms for EF1 allocations of chores and doubly monotone instances, and show that deciding the existence of an exactly envy-free chore allocation is NP-complete already for binary additive costs. Tao et al.~\cite{tao2025existence} study binary-marginal chores and prove the existence of an EFX and Pareto-optimal allocation for binary additive costs. The additive assumption cannot be dropped in general: Lin et al.~\cite{LLTZ26} construct binary XOS and binary supermodular instances for which no complete EFX allocation exists. In contrast, our result applies throughout the full binary-marginal chore domain: exact envy-freeness is obtained after subsidies, while the underlying allocation is always EF1.

Subsidies have also been studied directly for chore allocation. For proportionality under additive costs, tight bounds are known for equal entitlements~\cite{wu2025revisiting}, and weighted variants have been studied in~\cite{wu2024tree,GSW25}. More recently, Lu et al.~\cite{LMS26} show that for additive instances containing both goods and chores, one unit of subsidy per agent suffices, and the total bound of $n-1$ is tight; their model allows the sign of an item to differ across agents. Their result, however, does not subsume the present setting: negative dichotomous valuations need not be additive. Conversely, additive valuations need not be dichotomous. Thus, prior to the present work, the unit-subsidy guarantee was known for additive chores~\cite{LMS26} and for dichotomous goods~\cite{BKNS22}, but not for non-additive dichotomous chores. Our mixed-domain result is complementary: it permits non-additive valuations, while requiring every item to have a common sign across agents. Hence, the two results apply to incomparable domains.

Very recently, Dupr\'e la Tour and Suzuki~\cite{la2026subquadratic} obtained a subquadratic total-subsidy bound of $O(n^{3/2}\sqrt{\log n})$ for all nonpositive valuations with unit-bounded
marginals, a class containing general monotone chores. Their result is existential, and leaves open whether a linear total-subsidy bound is possible for general monotone valuations. In the negative dichotomous subclass studied here, we obtain the optimal linear bound of $n-1$, strengthen it to a per-agent subsidy in $\{0,1\}$, and provide a polynomial-time algorithm.

To the best of our knowledge, this is the first unit-subsidy guarantee for envy-freeness under negative dichotomous valuations.

\section{Notation and Preliminaries}
\label{sec:prelim}

We study the allocation of $m$ indivisible items among $n$ agents, with subsidies. Write $\agents = [n]$ for the set of agents and $\items$ for the set of items, $\abs{\items} = m$. Each agent $i \in \agents$ has a valuation $v_i : 2^{\items} \to \R$ with $v_i(\emptyset) = 0$, and we represent an instance by the tuple $\langle \agents, \items, \set{v_i}_{i \in \agents} \rangle$. Algorithms operate in the standard \emph{value-oracle} model: for each $v_i$ we assume only an oracle that returns $v_i(S)$ for a queried set $S \subseteq \items$, and running time is measured in $n$, $m$ and the number of oracle calls.

An \emph{allocation} $\alloc = (\bundle{1}, \dots, \bundle{n})$ is an ordered tuple of pairwise disjoint \emph{bundles}, with $\bundle{i}$ assigned to agent $i$. The allocation is \emph{complete} if $\bigcup_{i \in \agents} \bundle{i} = \items$ and \emph{partial} otherwise. The ordering matters throughout: much of the theory below concerns permuting a fixed family of bundles among the agents, and for a permutation $\sigma \in \mathbb{S}_n$ we write $\alloc_\sigma := (A_{\sigma(1)}, \dots, A_{\sigma(n)})$ for the allocation in which agent $i$ receives $A_{\sigma(i)}$. The utilitarian welfare of $\alloc$ is $\SW(\alloc) = \sum_{i \in \agents} v_i(\bundle{i})$.

\begin{definition}[Envy-Freeness]
\label{def:ef}
An allocation $\alloc = (\bundle{1},\dots,\bundle{n})$ is \emph{envy-free} (EF) iff $v_i(\bundle{i}) \ge v_i(\bundle{j})$ for all agents $i, j \in \agents$.
\end{definition}

Envy-free allocations of indivisible items need not exist, so we allow a third party to supplement the allocation with money. Agents have quasi-linear utilities: agent $i$ receiving bundle $\bundle{i}$ and subsidy $\subsidy_i$ enjoys utility $v_i(\bundle{i}) + \subsidy_i$, with money entering linearly and at the same exchange rate for every agent.

\begin{definition}[Envy-Free Solution] \label{def:efsolution}
An allocation $\alloc = (\bundle{1},\dots,\bundle{n})$ and a subsidy vector $\subsidy = (\subsidy_1,\dots,\subsidy_n) \in \R^n_{+}$ constitute an \emph{envy-free solution} $(\alloc, \subsidy)$ iff $v_i(\bundle{i}) + \subsidy_i \ge v_i(\bundle{j}) + \subsidy_j$ for all $i, j \in \agents$.
\end{definition}

An allocation $\alloc$ is \emph{envy-freeable} if some $\subsidy \in \R^n_+$ makes $(\alloc,\subsidy)$ an envy-free solution; this is a property of the allocation alone. 

We will also use the standard relaxation of envy-freeness up to one item (EF1). Since our framework allows chores, we use the two-sided form of EF1, in which one is allowed to remove at most one item from either of the two bundles being compared~\cite{kawase2025towards}. Formally:
\begin{definition}[EF1, two-sided form] \label{def:ef1}
An allocation $\alloc$ is \emph{envy-free up to one item} (EF1) iff for all $i, j \in \agents$ there exists $X \subseteq \bundle{i} \cup \bundle{j}$ with $\abs{X} \le 1$ such that $v_i(\bundle{i} \setminus X) \ge v_i(\bundle{j} \setminus X)$.
\end{definition}

For an all-goods instance, removing an item from the envious agent's own bundle can only make the comparison harder, so the definition reduces to the usual goods version of EF1, in which one item may be removed from the envied bundle~\cite{HS19}. For an all-chore instance, the situation is reversed: removing a chore from the envied bundle can only make that bundle more attractive. Hence, if an item is removed, it may, without loss, be taken from the envious agent's own bundle.

Apart from Section~\ref{sec:objective-mixed} and Appendices~\ref{app:goods-extension}--\ref{app:additive-objective-po}, every instance in this paper is an all-chore instance, and in this setting, Definition~\ref{def:ef1} takes a simpler form: the item removed can always be taken from the envious agent's own bundle. We record that specialisation, which is the standard chore version of EF1~\cite{aziz2019fair}, since it is the form that the results below establish.

\begin{definition}[EF1 for chores] \label{def:ef1chores}
A complete allocation $\alloc$ of chores is \emph{EF1} iff for all $i, j \in \agents$, either $i$ does not envy $j$, that is $v_i(\bundle i) \ge v_i(\bundle j)$, or there is a chore $e \in \bundle i$ with
\[
  v_i\bigl(\bundle i \setminus \set e\bigr) \;\ge\; v_i(\bundle j).
\]
\end{definition}

Definition~\ref{def:ef1chores} is exactly Definition~\ref{def:ef1} on an all-chore instance. For the converse, if the removed item lies in $A_j$, then
\[
    v_i(A_j\setminus\{e\})\geq v_i(A_j),
\]
so $v_i(A_i)\geq v_i(A_j\setminus\{e\})$ already implies $v_i(A_i)\geq v_i(A_j)$. Thus, whenever an item needs to be removed, it may be taken from the envious agent's own bundle.

\subsection{Negative Dichotomous Valuations}
\label{sec:negdich}

\begin{definition}[Negative dichotomous valuation] \label{def:negdich}
A valuation $v_i:2^{\items}\to\R$ with $v_i(\emptyset)=0$ is \emph{negative dichotomous} if, for every $S\subseteq\items$ and every $g\in\items\setminus S$,
\[
    \margplus{i}{S}{g}
    :=v_i(S\cup\set{g})-v_i(S)
    \in\set{-1,0}.
\]
Equivalently,
\[
    v_i(S)-v_i(S\cup\set{g})\in\set{0,1}.
\]
\end{definition}

Thus, adding an item weakly decreases an agent's valuation, so every item is a chore and $v_i$ is non-increasing:
\[
    S\subseteq T
    \quad\Longrightarrow\quad
    v_i(S)\ge v_i(T).
\]
Beyond the binary-marginal condition, we impose no further structure; in particular, the valuations need not be additive, submodular, or subadditive. This is the same level of generality considered by Barman et al.~\cite{BKNS22} for dichotomous goods.

For negative dichotomous valuations, it is often convenient to work with the corresponding non-negative \emph{cost function}
\[
    \cost_i(S):=-v_i(S).
\]
Then $\cost_i(\emptyset)=0$ and
\[
    \cost_i(S\cup\set{g})-\cost_i(S)\in\set{0,1}
\]
for every $S\subseteq\items$ and every $g\in\items\setminus S$.

\begin{definition}[Dichotomous cost function]
\label{def:dichcost}
A set function $\cost:2^{\items}\to\R$ is \emph{dichotomous} if $\cost(\emptyset)=0$ and
\[
    \cost(S\cup\set{g})-\cost(S)\in\set{0,1}
\]
for every $S\subseteq\items$ and every $g\in\items\setminus S$.
\end{definition}

Since all marginal costs are non-negative, every dichotomous cost function is automatically non-decreasing. Moreover, the correspondence
\[
    v_i \longleftrightarrow \cost_i=-v_i
\]
is a bijection between negative dichotomous valuations and dichotomous cost functions. In particular, the cost functions arising in our model belong to the same mathematical class as the dichotomous valuations studied by Barman et al.~\cite{BKNS22}; the difference is that here they represent costs rather than utilities.

In cost form, the envy-freeness condition for an allocation $\alloc=(\bundle1,\ldots,\bundle n)$ and subsidy vector $\subsidy$ becomes
\[
    \cost_i(\bundle i)-\subsidy_i
    \le
    \cost_i(\bundle j)-\subsidy_j
    \qquad\text{for all }i,j\in\agents.
\]
We use the valuation and cost representations interchangeably, according to which is more convenient.
\subsection{The Envy Graph and the Halpern--Shah Characterisation}
\label{sec:envygraph}

For an allocation $\alloc$, the \emph{envy graph} $\envygraph{\alloc}$ is the
complete weighted directed graph on the vertex set $\agents$ in which the arc
$(i,j)$ carries weight
\[
  \edgew{\alloc}(i,j) \;:=\; v_i(\bundle{j}) - v_i(\bundle{i})
  \;=\; \cost_i(\bundle{i}) - \cost_i(\bundle{j}),
\]
the envy that agent $i$ has towards agent $j$. The weight of a directed path $P$ is $\edgew{\alloc}(P) = \sum_{(i,j) \in P} \edgew{\alloc}(i,j)$, and $\pathw{\alloc}(i)$ denotes the maximum weight of any path in $\envygraph{\alloc}$ starting at $i$, the empty path of weight $0$ included.

The following two results of Halpern and Shah~\cite{HS19} underpin much of the literature on subsidies for envy-freeness. They apply to general valuation functions and hence, in particular, to the negative dichotomous valuations considered here. We state them in the notation of our model.
\begin{theorem}[\cite{HS19}] \label{thm:hs-characterisation}
For any allocation $\alloc = (\bundle{1},\dots,\bundle{n})$, the following are equivalent. \begin{enumerate}   \item[(i)] $\alloc$ is envy-freeable.   \item[(ii)] $\alloc$ maximises utilitarian welfare across all reassignments of its own bundles among the agents: for every permutation $\sigma$ over   $\agents$, $\sum_{i \in \agents} v_i(\bundle{i}) \ge    \sum_{i \in \agents} v_i(A_{\sigma(i)})$.   \item[(iii)] The envy graph $\envygraph{\alloc}$ has no positive-weight directed cycle.
\end{enumerate}
\end{theorem}

\begin{theorem}[\cite{HS19}]
\label{thm:hs-minsubsidy}
For any envy-freeable allocation $\alloc$, the subsidy vector $\optsubsidy$ given by $\optsubsidy_i := \pathw{\alloc}(i)$ realises an envy-free solution $(\alloc, \optsubsidy)$, and $\optsubsidy_i \le \subsidy_i$ for every $i \in \agents$ and every $\subsidy$ such that $(\alloc,\subsidy)$ is an envy-free solution. It can be computed in strongly polynomial time by running an all-pairs longest-path computation on $\envygraph{\alloc}$.
\end{theorem}

Condition~(ii) of Theorem~\ref{thm:hs-characterisation} implies that, starting from any collection of bundles, one can obtain an envy-freeable allocation by suitably reassigning those bundles among the agents. In particular, such a reassignment can be found by solving a maximum-weight perfect matching problem between agents and bundles. Accordingly, the main issue is not the existence of an envy-freeable allocation, but rather the magnitude of the subsidies required to support one, as characterised by Theorem~\ref{thm:hs-minsubsidy}.

For negative dichotomous valuations, all bundle values are integral: indeed, $v_i(S)$ is obtained by telescoping from $v_i(\emptyset)=0$ along marginals in $\{-1,0\}$. Consequently, every arc weight $\edgew{\alloc}(i,j)$ in the envy graph is an integer, and hence, by Theorem~\ref{thm:hs-minsubsidy}, the pointwise-minimal subsidy vector $\optsubsidy$ is integral as well. Therefore, showing $\optsubsidy_i\leq 1$ for every $i$ is equivalent to showing $\optsubsidy\in\{0,1\}^n$.

Moreover, every pointwise-minimal nonnegative subsidy vector has at least one zero coordinate. Indeed, if $\optsubsidy_i>0$ for every $i$, then subtracting $\min_i\optsubsidy_i$ from every coordinate preserves all envy-freeness inequalities while producing a strictly smaller nonnegative subsidy vector, contradicting pointwise minimality. Thus, if $\optsubsidy\in\{0,1\}^n$, then automatically
\[
    \sum_{i\in\agents}\optsubsidy_i\leq n-1.
\]

\section{Unit Subsidies for Negative Dichotomous Valuations}
\label{sec:main}

This section proves the main theorem: for any number of agents $n$ and every negative dichotomous instance, there is an allocation admitting an envy-free solution with a subsidy of $0$ or $1$ per agent, at most $n-1$ in total, computable in polynomial time.

Our argument starts from an envy-free \emph{partial} allocation produced by the algorithm of Tao, Wu, Yu, and Zhou~\cite{tao2025existence}, and then completes this partial allocation. We recall in \Cref{sec:m-terminal} the relevant terminal-state properties of their algorithm, including Theorem~6.1 and the three update rules used in its construction. These ingredients are due to Tao et al.~\cite{tao2025existence}. The completion procedure introduced in \Cref{sec:m-completion}, together with the subsequent lemmas and theorems, constitutes our contribution.

We work in cost form throughout, so an envy-free solution (Definition~\ref{def:efsolution}) is the requirement $(\mathcal{A},p)$ such that
\begin{equation} \label{eq:m-ef}
  \cost_i(\bundle i) - \subsidy_i \;\le\; \cost_i(\bundle j) - \subsidy_j  \qquad \text{for all } i,j \in \agents ,
\end{equation} 
and a partial allocation $X$ is envy-free (Definition~\ref{def:ef}) when
\begin{equation} \label{eq:m-pef}
  \cost_i(X_i) \;\le\; \cost_i(X_j)
  \qquad \text{for all } i,j \in \agents .
\end{equation}
For $S \subseteq \items$ and $e \in \items \setminus S$ write $\cost_i(e \mid S) := \cost_i(S \cup \set e) - \cost_i(S) \in \set{0,1}$ for a marginal, and for a partial allocation $X$ write $R(X) := \items \setminus \bigcup_{i \in \agents} X_i$ for its unallocated chores.  Moreover, as $\cost_i$s are integer valued, for all $S,T\subseteq M$, \(\cost_i(S) < \cost_i(T) \implies \cost_i(S) \le \cost_i(T) - 1 .
 \)

To describe the algorithm of Tao, Wu, Yu, and Zhou~\cite{tao2025existence} and our completion procedure, we use the following standard graph-theoretic notions.

\begin{definition}[Strongly connected component and tail component]
\label{def:scc}
Let $G=(V,E)$ be a finite directed graph. Two vertices \(u,v\in V\) are \emph{strongly connected} if either \(u=v\), or there is a directed path from \(u\) to \(v\) and a directed path from \(v\) to \(u\). Strong connectivity is an equivalence relation on $V$, and its equivalence classes are the \emph{strongly connected components} of $G$.

A strongly connected component $S$ is a \emph{tail component} (or \emph{sink component}) if no arc leaves $S$; that is, there is no $(i,j)\in E$ with $i\in S$ and $j\in V\setminus S$.
\end{definition}

The following standard graph-theoretic fact will be used repeatedly.

\begin{lemma} \label{lem:tail-scc-exists} Every finite non-empty directed graph has at least one tail strongly connected component.
\end{lemma}

\subsection{The terminal state of the partial-allocation algorithm}
\label{sec:m-terminal}

\begin{theorem}[\cite{tao2025existence}, Theorem~6.1] \label{thm:m-twyz} For every $n$ and all dichotomous cost functions $\cost_1,\dots,\cost_n$ there is a polynomial-time algorithm returning a partial allocation $X$ that is envy-free in the sense of \eqref{eq:m-pef} and satisfies $\abs{R(X)} \le n-1$. \end{theorem}

The proof of our main result uses not only the statement of Theorem~\ref{thm:m-twyz}, but also structural properties of the terminal state of the algorithm of~\cite{tao2025existence}. We therefore recall the relevant part of that algorithm. It maintains an envy-free partial allocation $X=(X_1,\dots,X_n)$ together with the \emph{equality graph}
\begin{equation}
\label{eq:m-graph}
  G=(\agents,E),
  \qquad
  (i,j)\in E
  \iff
  i\ne j \ \text{ and }\ \cost_i(X_i)=\cost_i(X_j).
\end{equation}
While $R(X)\ne\emptyset$, the algorithm considers the following rules in the stated order.
\begin{enumerate}
  \item[(R1)] If there exist $e\in R(X)$ and $i\in\agents$ such that
        $\cost_i(e\mid X_i)=0$, assign $e$ to $i$.

  \item[(R2)] Otherwise, if there exist $e\in R(X)$ and an arc      $(i,j)\in E$ lying on a directed cycle of $G$ such that $\cost_i(e\mid X_j)=0$, rotate the bundles along that cycle and assign $e$ to $i$.

 \item[(R3)] Otherwise, select a tail strongly connected component $S$ of $G$. If $\abs{R(X)}\geq\abs S$, assign distinct residual chores to the agents of $S$, one chore per agent. Otherwise, halt. \end{enumerate}           Consequently, if the algorithm terminates with \(R(X)\neq\emptyset\), then neither (R1) nor (R2) is applicable at the terminal state. Moreover, if \(S\) denotes the tail strongly connected component selected in the terminating application of (R3), then
\[
    \abs{R(X)}<\abs S.
\]

For the remainder of this section, let $X$ denote the terminal partial allocation of the algorithm, and write
\[
    R:=R(X), \qquad r:=\abs R.
\]
Unless stated otherwise, let $G=(\agents,E)$ denote the equality graph defined in~\eqref{eq:m-graph} at this terminal state. When \(r\ge1\), let \(S\) be the tail SCC selected at the terminating application of (R3); therefore \(r<|S|\).

\begin{lemma}
\label{lem:m-terminal}
Suppose $r \ge 1$. Then
\begin{enumerate}
  \item[(i)] $\cost_i(e \mid X_i) = 1$ for every $i \in \agents$ and every $e \in R$;
  \item[(ii)] $\cost_i(e \mid X_j) = 1$ for every arc $(i,j) \in E$ with  $i,j \in S$ and every $e \in R$;
  \item[(iii)] $\cost_i(X_i) < \cost_i(X_j)$ for every $i \in S$ and every $j \in \agents \setminus S$.
\end{enumerate}
\end{lemma}

\begin{proof} Since $r\geq 1$, the algorithm terminates with a nonempty set of unallocated chores. Therefore, neither \textnormal{(R1)} nor \textnormal{(R2)} is applicable at the terminal state.

(i) If $\cost_i(e\mid X_i)=0$ for some $i\in\agents$ and $e\in R$, then \textnormal{(R1)} would be applicable. Since every marginal cost belongs to $\set{0,1}$ by Definition~\ref{def:dichcost}, it follows that
\[
    \cost_i(e\mid X_i)=1.
\]

(ii) Let $(i,j)\in E$ with $i,j\in S$. Since $S$ is strongly connected, there is a directed path from $j$ to $i$. Choosing such a path with no repeated vertices and appending the arc $(i,j)$ yields a directed cycle containing $(i,j)$. If $\cost_i(e\mid X_j)=0$ for some $e\in R$, then \textnormal{(R2)} would be applicable to this arc and chore. Hence $\cost_i(e\mid X_j)\neq 0$, and therefore
\[
    \cost_i(e\mid X_j)=1.
\]

(iii) Let $i\in S$ and $j\in\agents\setminus S$. Since $S$ is a tail strongly connected component, no arc leaves $S$; in particular, $(i,j)\notin E$. By~\eqref{eq:m-graph},
\[
    \cost_i(X_i)\neq \cost_i(X_j).
\]
On the other hand, envy-freeness of $X$, as expressed in \eqref{eq:m-pef}, gives
\[
    \cost_i(X_i)\leq \cost_i(X_j).
\]
The inequality must therefore be strict.
\end{proof}
\subsection{The completion}
\label{sec:m-completion}
We complete the partial allocation $X$ to $\mathcal{A}$ by allocating the remaining chores to agents in $S$ so that $\mathcal{A}$ becomes envy-freeable with a subsidy vector $p\in\{0,1\}^n$. If $r=0$, then $X$ is already complete. Hence, assume that $r\geq 1$. We first give some intuition for the completion procedure before stating it formally.

The completion cannot be chosen arbitrarily. There are two potential difficulties. First, by Lemma~\ref{lem:m-terminal}(i), every residual chore raises its recipient's own cost by one. Concentrating two such chores on an agent who was tied with another unchanged bundle can therefore require a subsidy difference of two, which is impossible with subsidies in $\{0,1\}^n$. We consequently assign the residual chores to distinct recipients.

Second, even after giving only one residual chore to each recipient, subsidising the recipients alone need not suffice. Suppose that an agent $i$ has an equality edge $(i,t)$ to a recipient $t$, so that
\[
c_i(X_i)=c_i(X_t),
\]
and that the residual chore $e$ assigned to $t$ has zero marginal cost for $i$ on $X_t$. Paying one unit to $t$ then gives
\[
c_i(X_t\cup\{e\})-1
   =c_i(X_i)-1,
\]
so agent $i$ would envy $t$ unless $i$ were subsidised as well. The same phenomenon may then propagate backwards through further equality edges.

The tail component $S$ is precisely what controls this propagation. By Lemma~\ref{lem:m-terminal}(ii), for every equality edge $(i,t)$ with $i,t\in S$ and every residual chore $e$, we have $c_i(e\mid X_t)=1$. Hence, if $t\in S$ receives $e$ and one unit of subsidy, then
\[
c_i(X_t\cup\{e\})-1
   =c_i(X_t)
   =c_i(X_i).
\]
Thus, from the viewpoint of an agent in $S$ who was tied with the recipient, the additional chore exactly offsets the recipient's unit subsidy, so no further subsidy is needed within $S$ for this reason.

Moreover, since $S$ is a tail component, there is no equality edge from an agent in $S$ to an agent outside $S$. Hence, once the backward propagation reaches an agent outside $S$, any further predecessor that it forces into the subsidy set must also lie outside $S$. Consequently, the only subsidised agents in $S$ are the recipients. Since $r<|S|$, assigning the $r$ residual chores to distinct agents in $S$ therefore leaves at least one agent in $S$ unsubsidised.

We now formalize this idea. Assume $r\geq 1$ and write
\(
    R=\set{e_1,\dots,e_r}.
\) Since $r<\abs S$, we may choose $r$ distinct agents \( T=\set{t_1,\dots,t_r}\subseteq S \) arbitrarily. Define the allocation $\alloc=(A_1,\dots,A_n)$ by
\begin{equation}
\label{eq:m-alloc}
  A_i :=
  \begin{cases}
    X_{t_k}\cup\set{e_k}, & i=t_k \quad (1\leq k\leq r),\\[2pt]
    X_i, & i\in\agents\setminus T.
  \end{cases}
\end{equation}
Since the agents $t_1,\dots,t_r$ are distinct and each chore in $R$ is assigned exactly once, the bundles in $\alloc$ remain pairwise disjoint and their union is $\items$. Hence, $\alloc$ is a complete allocation. We refer to the agents in $T$ as the \emph{recipients}.

We next construct a subsidy vector $\subsidy\in\set{0,1}^n$ and show that $(\alloc,\subsidy)$ is envy-free. The following definition specifies the set of agents who receive one unit of subsidy.

\begin{definition}[Subsidy set]
\label{def:m-P}
Let $P\subseteq\agents$ be the smallest set satisfying
\begin{enumerate}
  \item[(P1)] $T\subseteq P$; and
  \item[(P2)] if $i\in\agents\setminus S$ and $(i,j)\in E$ for some
        $j\in P$, then $i\in P$.
\end{enumerate}
Equivalently, initialise $P:=T$ and repeatedly add any agent $i\in\agents\setminus S$ for which $(i,j)\in E$ for some $j$ already in $P$. Continue until no further agents can be added. Since $\agents$ is finite and the set $P$ only grows, this procedure terminates.
Define
\begin{equation}
\label{eq:m-subsidy}
  \subsidy_i :=
  \begin{cases}
    1, & i\in P,\\
    0, & i\in\agents\setminus P.
  \end{cases}
\end{equation}
\end{definition}

The complete procedure is summarised in Algorithm~\ref{alg:main}. The partial-allocation phase in Line~2 is the partial-allocation algorithm of Tao et al.~\cite{tao2025existence}; the completion and subsidy construction in the subsequent lines constitute our contribution.

\begin{algorithm}[H]
\caption{EF1 allocation with at most one unit of subsidy per agent for negative dichotomous valuations}
\label{alg:main}
\begin{algorithmic}[1]
\Require Agents $N=[n]$, chores $M$, and value-oracle access to negative
dichotomous valuations $\{v_i\}_{i\in N}$.
\Ensure A complete allocation $\mathcal{A}$ that is EF1, together with a
subsidy vector $p\in\{0,1\}^n$ such that $(\mathcal{A},p)$ is envy-free and
$\sum_{i\in N}p_i\le n-1$.

\State Set $c_i=-v_i$ for every $i\in N$.
\State Run the partial-allocation algorithm of Tao et al.~\cite{tao2025existence} on $\set{\cost_i}_{i\in\agents}$, retaining the tail strongly connected component selected in the terminating application of  \textnormal{(R3)} if the set of unallocated chores is nonempty.
\State Let $X=(X_1,\dots,X_n)$ be its terminal partial allocation and let
\[
R=\items\setminus\bigcup_{i\in\agents}X_i .
       \]
\If{$R=\emptyset$}
    \State \Return $(X,\mathbf 0)$.
\EndIf

\State Construct the terminal equality graph $G=(\agents,E)$, where
       \[
       (i,j)\in E
       \iff
       i\neq j
       \text{ and }
       \cost_i(X_i)=\cost_i(X_j).
       \]
\State Write $R=\{e_1,\ldots,e_r\}$.
\State Let $S$ be the tail strongly connected component selected in the
       terminating application of \textnormal{(R3)}.
       \Comment{$r<\abs S$}
\State Choose arbitrary distinct agents
      $T=\{t_1,\ldots,t_r\}\subseteq S$.
      \Comment{possible since $r<|S|$}

\For{$k=1,\ldots,r$}
    \State $A_{t_k}\gets X_{t_k}\cup\{e_k\}$.
\EndFor
\For{$i\in N\setminus T$}
    \State $A_i\gets X_i$.
\EndFor

\State Initialise $P\gets T$.
\While{there exist $i\in N\setminus S$ and $j\in P$
       with $(i,j)\in E$ and $i\notin P$}
    \State $P\gets P\cup\{i\}$.
\EndWhile

\For{$i\in N$}
    \State $p_i\gets \mathbf{1}[i\in P]$.
\EndFor

\State $\mathcal{A}\gets(A_1,\ldots,A_n)$.
\State \Return $(\mathcal{A},p)$.
\end{algorithmic}
\end{algorithm}

\begin{lemma}
\label{lem:m-Pstructure}
$P\cap S=T$. Equivalently,
\[P\setminus T\subseteq \agents\setminus S.
\] Thus, every subsidised agent in $S$ is a recipient.
\end{lemma}

\begin{proof}
Since $T\subseteq P$ by \textnormal{(P1)} and $T\subseteq S$, we have $T\subseteq P\cap S$. Suppose, for contradiction, that there exists $i\in (P\cap S)\setminus T$. Consider
\[
    P':=P\setminus\set{i}.
\]
Since $i\notin T$, the set $P'$ still satisfies \textnormal{(P1)}. We claim that $P'$ also satisfies \textnormal{(P2)}. Let $k\in\agents\setminus S$ and suppose that $(k,j)\in E$ for some $j\in P'$. Since $P'\subseteq P$ and $P$ satisfies \textnormal{(P2)}, we have $k\in P$. Moreover, $k\neq i$, because $k\notin S$ whereas $i\in S$. Hence $k\in P'$.

Thus, $P'$ satisfies both \textnormal{(P1)} and \textnormal{(P2)}, contradicting the minimality of $P$. Therefore $P\cap S=T$.
\end{proof}

By \eqref{eq:m-subsidy}, $\subsidy\in\set{0,1}^n$. Moreover, Lemma~\ref{lem:m-Pstructure} gives $P\cap S=T$, and hence $P\setminus T\subseteq \agents\setminus S$. Therefore,
\begin{equation}
\label{eq_1}
\sum_{i\in N}p_i
 = |P|
 = |T|+|P\setminus T|
 \le r+n-|S|
 = n-(|S|-r)
 \le n-1,
\end{equation}
where the last inequality follows from $r<|S|$. Thus, more precisely, the construction leaves at least $|S|-r$ agents unsubsidised.

\begin{lemma}
\label{lem:m-expensive}
Let $i\in S$, $t\in T$, and $e\in R$. Then
\[
    \cost_i(X_t\cup\set e)\geq \cost_i(X_i)+1.
\]
\end{lemma}

\begin{proof}
Since $X$ is envy-free, \( \cost_i(X_i)\leq \cost_i(X_t). \) If the inequality is strict, the integrality of the cost functions gives \( \cost_i(X_t)\geq \cost_i(X_i)+1, \) and the desired inequality follows from the monotonicity of $\cost_i$. Suppose instead that \(   \cost_i(X_i)=\cost_i(X_t). \) If $i=t$, then Lemma~\ref{lem:m-terminal}(i) gives
\[
    \cost_i(e\mid X_i)=1.
\]
If $i\neq t$, then $(i,t)\in E$ by \eqref{eq:m-graph}; since $i,t\in S$, Lemma~\ref{lem:m-terminal}(ii) gives
\[
    \cost_i(e\mid X_t)=1.
\] In either case,
\[
    \cost_i(X_t\cup\set e)=\cost_i(X_i)+1.
\]
\end{proof}
\subsection{Envy-freeness of the completion}
\label{sec:m-ef}
\begin{theorem}
\label{thm:m-completion}
Let $X$, $R$, and $S$ be as in \Cref{sec:m-terminal} with $r\geq 1$, let $\alloc$ be defined by~\eqref{eq:m-alloc}, and let $\subsidy$ be defined by~\eqref{eq:m-subsidy}. Then $(\alloc,\subsidy)$ is an envy-free solution. \end{theorem}

\begin{proof}
Fix $i\in\agents$. Since $X_j\subseteq A_j$ for every $j\in\agents$, monotonicity of $\cost_i$ and envy-freeness of $X$ give
\begin{equation}
\label{eq:m-mono}
    \cost_i(A_j)
    \geq \cost_i(X_j)
    \geq \cost_i(X_i)
    \qquad\text{for every }j\in\agents.
\end{equation}
We verify
\(\cost_i(A_i)-\subsidy_i \leq \cost_i(A_j)-\subsidy_j \) for every $j\in\agents$.

\smallskip
\noindent\textbf{Case 1: $i\in T$.}
Then $\subsidy_i=1$, and if $e\in R$ is the residual chore assigned to $i$, Lemma~\ref{lem:m-terminal}(i) gives
\[
    \cost_i(A_i)
    =\cost_i(X_i)+\cost_i(e\mid X_i)
    =\cost_i(X_i)+1.
\]
Hence $\cost_i(A_i)-\subsidy_i=\cost_i(X_i)$. If $j\notin P$, then $j\notin T$, so $A_j=X_j$ and $\subsidy_j=0$. Thus, by envy-freeness of $X$,
\[
    \cost_i(A_j)-\subsidy_j
    =\cost_i(X_j)
    \geq \cost_i(X_i).
\] Suppose instead that $j\in P$. If $j\in T$, then Lemma~\ref{lem:m-expensive} gives \( \cost_i(A_j)\geq \cost_i(X_i)+1. \) If $j\in P\setminus T$, then Lemma~\ref{lem:m-Pstructure} implies $j\notin S$. Since $i\in S$, Lemma~\ref{lem:m-terminal}(iii) gives \(\cost_i(X_i)<\cost_i(X_j).\) By integrality, \(  \cost_i(A_j)=\cost_i(X_j)\geq \cost_i(X_i)+1. \) Thus, in either case, since $\subsidy_j=1$,
\[
    \cost_i(A_j)-\subsidy_j\geq \cost_i(X_i).
\]

\smallskip
\noindent\textbf{Case 2: $i\in P\setminus T$.} Here $A_i=X_i$ and $\subsidy_i=1$, so \(\cost_i(A_i)-\subsidy_i=\cost_i(X_i)-1. \) For every $j\in\agents$, \eqref{eq:m-mono} and $\subsidy_j\leq 1$ imply
\[
    \cost_i(A_j)-\subsidy_j
    \geq \cost_i(X_i)-1
    =\cost_i(A_i)-\subsidy_i.
\]

\smallskip
\noindent\textbf{Case 3: $i\notin P$.}
Since $T\subseteq P$, we have $i\notin T$, and therefore $A_i=X_i$ and $\subsidy_i=0$. Hence \(\cost_i(A_i)-\subsidy_i=\cost_i(X_i). \) If $j\notin P$, then $A_j=X_j$ and $\subsidy_j=0$, and envy-freeness of $X$ gives
\[
    \cost_i(A_j)-\subsidy_j
    =\cost_i(X_j)
    \geq\cost_i(X_i)=\cost_i(A_i)-\subsidy_i.
\]
Now suppose that $j\in P$. We claim that
\begin{equation}
\label{eq:m-unpaid-to-paid}
    \cost_i(A_j)\geq \cost_i(X_i)+1.
\end{equation}
If $i\notin S$, envy-freeness of $X$ gives $\cost_i(X_i)\leq\cost_i(X_j)$. Equality would imply $(i,j)\in E$ by~\eqref{eq:m-graph}; since $i\in\agents\setminus S$ and $j\in P$, condition \textnormal{(P2)} would then imply $i\in P$, a contradiction. Hence \(\cost_i(X_i)<\cost_i(X_j), \) and integrality together with monotonicity yield
\[
\cost_i(A_j)\geq\cost_i(X_j)\geq\cost_i(X_i)+1.
\] It remains to consider $i\in S$. If $j\in T$, then \eqref{eq:m-unpaid-to-paid} follows from Lemma~\ref{lem:m-expensive}. If $j\in P\setminus T$, then Lemma~\ref{lem:m-Pstructure} gives $j\notin S$, and Lemma~\ref{lem:m-terminal}(iii), together with integrality, gives
\[
    \cost_i(A_j)=\cost_i(X_j)\geq\cost_i(X_i)+1.
\]
Thus~\eqref{eq:m-unpaid-to-paid} holds in all cases. Since $\subsidy_j=1$,
\[
    \cost_i(A_j)-\subsidy_j
    \geq\cost_i(X_i)
    =\cost_i(A_i)-\subsidy_i.
\]
The envy-freeness inequality therefore holds for every pair $i,j\in\agents$. \end{proof}

\subsection{The EF1 property of the completion}

\begin{corollary}
\label{cor:ef1}
The complete allocation $\mathcal{A}$ constructed in Section~\ref{sec:m-completion} is EF1 in the sense of Definition~\ref{def:ef1chores}, and hence also in the sense of Definition~\ref{def:ef1}. More precisely:
\begin{enumerate}
    \item [(i)] every non-recipient $i\in N\setminus T$ is envy-free in  $\mathcal{A}$; and
    \item [(ii)] for every recipient $t_k\in T$,
    \[ c_{t_k}(A_{t_k}\setminus\{e_k\}) \leq c_{t_k}(A_j) \qquad\text{for every }j\in N.
    \]
\end{enumerate}
\end{corollary}

\begin{proof}
Recall that the partial allocation $X=(X_1,\ldots,X_n)$ is envy-free and that $X_j\subseteq A_j$ for every $j\in N$. Let first $i\in N\setminus T$. Then $A_i=X_i$. Hence, for every $j\in N$, envy-freeness of $X$ and monotonicity of $c_i$ give
\[
    c_i(A_i)
    =
    c_i(X_i)
    \leq
    c_i(X_j)
    \leq
    c_i(A_j).
\]
Thus, every non-recipient is envy-free in $\mathcal{A}$. Now let $i=t_k\in T$. By construction,
\[ A_i=X_i\cup\{e_k\},
    \qquad\text{and hence}\qquad
    A_i\setminus\{e_k\}=X_i.
\]
Therefore, for every $j\in N$,
\[c_i(A_i\setminus\{e_k\})
    =
    c_i(X_i)
    \leq
    c_i(X_j)
    \leq
    c_i(A_j),
\]
where the two inequalities follow from envy-freeness of $X$ and monotonicity of $c_i$, respectively. Thus, removing the residual chore assigned to a recipient eliminates all of her envy. Consequently, $\mathcal{A}$ is EF1.
\end{proof}
\subsection{The main theorem}
\label{sec:m-main}

\begin{theorem}
\label{thm:m-main}
For every instance
$\langle \agents,\items,\set{v_i}_{i\in\agents}\rangle$ with negative dichotomous valuations, there exist a complete allocation $\alloc$ and a subsidy vector $\subsidy\in\set{0,1}^n$ such that
\[
    \sum_{i\in\agents}\subsidy_i\leq n-1
\]
and $(\alloc,\subsidy)$ is an envy-free solution. Moreover, $\alloc$ is EF1 even before subsidies are paid. Such an allocation and subsidy vector can be computed in polynomial time given value-oracle access.
\end{theorem}

\begin{proof}
Pass to the cost representation $\cost_i=-v_i$, which is dichotomous in the sense of Definition~\ref{def:dichcost}. Run the partial-allocation algorithm of Tao et al.~\cite{tao2025existence}, and let $X$ be its terminal partial allocation, with set of unallocated chores
\[ R:=R(X), \qquad r:=\abs R.
\]

If $r=0$, then $X$ is complete and satisfies~\eqref{eq:m-pef}. Hence $(X,\mathbf 0)$ is an envy-free solution. Since every envy-free allocation is EF1, $X$ is also EF1.

If $r\geq 1$, let $G=(\agents,E)$ be the terminal equality graph and let $S$ be the tail strongly connected component of $G$ fixed in \Cref{sec:m-terminal}. Construct $T$, $\alloc$, and $\subsidy$ as in \Cref{sec:m-completion}. Then $\alloc$ is complete, Theorem~\ref{thm:m-completion} shows that $(\alloc,\subsidy)$ is envy-free, \eqref{eq:m-subsidy} and \eqref{eq_1} give \[
    \subsidy\in\set{0,1}^n
    \qquad\text{and}\qquad
    \sum_{i\in\agents}\subsidy_i\leq n-1,
\] and Corollary~\ref{cor:ef1} shows that $\alloc$ is EF1.

It remains to verify polynomial-time computability. The partial-allocation algorithm of Tao et al.~\cite{tao2025existence} runs in polynomial time. If $r\geq 1$, the terminal equality graph has $n$ vertices and can be constructed from the $n^2$ values $\cost_i(X_j)$, requiring one value-oracle query for each pair $(i,j)$. The tail component $S$ selected in the terminating application of \textnormal{(R3)} is retained during the partial-allocation phase. By Theorem~\ref{thm:m-twyz}, $r\leq n-1$, so constructing $\alloc$ requires at most $n-1$ additional assignments. Finally, $P$ can be computed by a backward reachability search from $T$, restricted to vertices outside $S$, in $O(n^2)$ time. Hence, the entire procedure runs in polynomial time. \end{proof}

As noted in the introduction, the total-subsidy bound in Theorem~\ref{thm:m-main} is tight: there are instances for which every complete envy-free solution requires total subsidy at least \(n-1\).

\subsection{Pareto optimality and subsidies}
We next examine whether the unit-subsidy guarantee can be obtained together with Pareto optimality of the underlying allocation. Pareto optimality is evaluated with respect to the allocation of chores, without taking subsidies into account. Formally, a complete allocation $\allocB$ \emph{Pareto dominates} a complete allocation $\alloc$ if
\[\cost_i(\allocB_i)\leq \cost_i(\bundle i)
    \qquad\text{for every }i\in\agents,
\]
with strict inequality for at least one agent. A complete allocation is \emph{Pareto optimal} if it is not Pareto dominated by any other complete allocation.

For the binary additive subclass, Pareto optimality is compatible with the unit-subsidy guarantee. This parallels the corresponding phenomenon for binary additive goods, where Halpern and Shah~\cite{HS19} obtain a Pareto-efficient allocation requiring total subsidy at most $n-1$; see also Goko et al.~\cite{goko2024fair} for unit subsidies together with utilitarian optimality in the more general matroidal goods setting. For completeness, we record the simple chore-side observation.

\begin{proposition}
\label{prop:additive-po}
For binary additive costs, there exists a complete allocation $\mathcal A$ and $p\in\{0,1\}^n$ such that $(\mathcal A,p)$ is envy-free, $\sum_{i\in N}p_i\le n-1$, and $\mathcal A$ is Pareto optimal and EF1. Moreover, such an outcome can be computed in polynomial time.
\end{proposition}

\begin{proof}
Let
\[
U:=\{e\in M:c_i(\{e\})=1\text{ for every }i\in N\}
\]
be the universally costly chores. Assign every $e\notin U$ to an agent for whom $c_i(e)=0$. Write $|U|=qn+r$, where $0\le r<n$, and distribute the chores in $U$ so that $r$ agents receive $q+1$ of them and the remaining agents receive $q$.

Set $p_i=1$ for the former $r$ agents and $p_i=0$ for the others. Then $c_i(A_i)-p_i=q$ for every $i$. Moreover, every bundle $A_j$ contains $q+p_j$ universally costly chores, and hence
\[
c_i(A_j)-p_j\ge q=c_i(A_i)-p_i
\]
for all $i,j$. Thus $(\mathcal A,p)$ is envy-free and $\sum_i p_i=r\le n-1$.

Every chore outside $U$ is allocated at zero cost, whereas every chore in $U$ necessarily contributes one unit of cost under every allocation. Hence $\mathcal A$ minimizes total cost and is therefore Pareto optimal. Finally, an agent receiving $q$ universally costly chores is already envy-free, while an agent receiving $q+1$ becomes envy-free after
removing one universally costly chore; thus $\mathcal A$ is EF1. The construction clearly requires only polynomially many value queries
and operations.
\end{proof}

Thus, in the binary additive domain, Pareto optimality is compatible with the optimal unit-subsidy guarantee. The total-subsidy bound remains tight even under Pareto optimality: in the identical binary additive instance $c_i(S)=|S|$ from the introduction, every complete allocation is Pareto optimal, while every envy-free solution requires total subsidy at least $n-1$. This compatibility, however, breaks down once nonadditivity is allowed, even for identical binary submodular costs.

\begin{proposition}
\label{prop:po-impossibility}
There exist negative dichotomous instances, even with two agents and identical binary submodular cost functions, for which no Pareto-optimal allocation admits an envy-free subsidy vector in $\{0,1\}^n$.
\end{proposition}

\begin{proof}
Consider two agents $\agents=\set{1,2}$ and three chores $\items=\set{a,b,c}$. Both agents have the same cost function
\[ \cost_1(S)=\cost_2(S)=\min\{\abs S,2\}.
\]
Every marginal cost belongs to $\{0,1\}$, so the cost function is dichotomous. Moreover, it is submodular, since $k\mapsto\min\{k,2\}$ is a nondecreasing concave function of cardinality. It is not additive, since
\[\cost_i(\set{a,b,c})=2<3
    =\sum_{g\in\set{a,b,c}}\cost_i(\set g).
\]
Every complete allocation has, up to symmetry, either bundle-size profile $(3,0)$ or $(2,1)$. First, consider the allocation in which agent $1$ receives all three chores. Its cost profile is $(2,0)$. The possible cost profiles of complete allocations are
\[(2,0),\qquad (2,1),\qquad (1,2),\qquad (0,2),
\]
and none of the latter three Pareto dominates $(2,0)$. Hence, the concentrated allocation is Pareto optimal. By symmetry, the allocation in which agent $2$ receives all three chores is also Pareto optimal. However, a concentrated allocation cannot be made envy-free with unit subsidies. If agent $1$ receives all three chores, envy-freeness requires
\[
    2-p_1\leq -p_2,
\] and hence
\[
    p_1-p_2\geq 2.
\]
Thus, no subsidy vector in $\{0,1\}^2$ suffices. The case in which agent $2$ receives all three chores is symmetric.

Now consider a split allocation, with bundle sizes $(2,1)$ up to symmetry. Its cost profile is $(2,1)$, so it can be made envy-free with subsidies $(1,0)$, up to symmetry. Nevertheless, it is not Pareto optimal. Transferring the singleton chore to the agent holding two chores changes the cost profile from $(2,1)$ to $(2,0)$: the cost of one agent remains unchanged, while the other agent's cost strictly decreases.

Thus, the Pareto-optimal allocations are precisely the two concentrated allocations, and each requires a subsidy of at least $2$ for the agent receiving all three chores. Conversely, every allocation that admits an envy-free subsidy vector in $\{0,1\}^2$ is Pareto dominated. Therefore, no allocation is simultaneously Pareto optimal and envy-free with subsidies in $\{0,1\}^2$.
\end{proof}

Thus, the unit-subsidy guarantee of Theorem~\ref{thm:m-main} cannot, in general, be strengthened to require Pareto optimality.

\section{Unit Subsidies for Objective Signed Dichotomous Valuations}
\label{sec:objective-mixed}

We now consider mixed instances in which every item is commonly classified as a good or a chore. Such common-sign settings are referred to as \emph{objective} goods and chores or \emph{objective valuations} in the fair-division literature \cite{aziz2022fair,bilo2026approximately}. We additionally impose signed binary marginals.

Formally, a valuation profile is \emph{objective signed dichotomous} if the item set admits a common partition
\[
M=G\mathbin{\dot\cup}C
\]such that, for every agent $i\in N$ and every $S\subseteq M$,
\[
v_i(S\cup\{g\})-v_i(S)\in\{0,1\}
\qquad
\text{for every }g\in G\setminus S,
\]
and
\[
v_i(S\cup\{c\})-v_i(S)\in\{0,-1\}
\qquad
\text{for every }c\in C\setminus S.
\]
Thus, every item is commonly regarded as either a good or a chore, although its marginal value may depend on the bundle to which it is added. In particular, the valuations need not be additive, submodular, or subadditive. It is a subclass of doubly monotone valuations in which the classification of items into goods and chores is common across agents.

The key step is to show that the one-good extension argument of Barman et al.~\cite{BKNS22} remains valid even when the already allocated bundles may contain chores.

\begin{lemma}
\label{lem:goods-extension}
Let $M'$ be a set of allocated items and let $A=(A_1,\ldots,A_n)$ be a complete allocation of $M'$. Suppose that the valuations are integer-valued and that there exists $p\in\{0,1\}^n$ such that $(A,p)$ is envy-free. Let $g\notin M'$
satisfy
\[
v_i(S\cup\{g\})-v_i(S)\in\{0,1\}
\]
for every agent $i$ and every $S\subseteq M'$. Then there exist a complete allocation $A'$ of $M'\cup\{g\}$ and $p'\in\{0,1\}^n$ such that $(A',p')$ is envy-free. Moreover, $(A',p')$ can be computed in polynomial time given value-oracle access.
\end{lemma}

\begin{proof}
We give a proof sketch here; the full proof is given in Appendix~\ref{app:goods-extension}. For an allocation $Y$, let $G_Y$ have edge weights $w_Y(i,j):=v_i(Y_j)-v_i(Y_i)$. We use the fact that $Y$ is envy-freeable iff it maximizes welfare over all bundle permutations (equivalently, $G_Y$ has no positive-weight cycle), and that an agent's pointwise-minimal subsidy is the maximum weight of a directed path starting at that agent~\cite{HS19} (Theorems~\ref{thm:hs-characterisation} and~\ref{thm:hs-minsubsidy}). We may replace $p$ by this pointwise-minimal vector, which remains in $\{0,1\}^n$.

Two facts are crucial. If $Y=(Y_1,\ldots,Y_n)$ is envy-freeable and $g$ is assigned to an agent $x$ with marginal value one, the new allocation remains envy-freeable: its welfare rises by one, while any permutation can gain at most one from $g$. Moreover, if $Z$ is obtained by adding $g$ to $Y_x$, then
\[
w_Z(x,j)\leq w_Y(x,j),\qquad w_Z(i,x)\leq w_Y(i,x)+1,
\]
with all other edges unchanged; hence every simple path gains at most one in weight.

We now use the extendable and non-extendable cases of Barman et al.~\cite{BKNS22}. In the extendable case, after a feasible bundle permutation, $g$ is assigned to a maximally subsidized agent with marginal value one. The two facts above preserve envy-freeability and keep the required subsidies in $\{0,1\}^n$; if the recipient has subsidy one, it can be removed after adding $g$.

In the non-extendable case, we first show that allocating the new good $g$ to any maximally subsidized agent yields an envy-freeable allocation. Accordingly, \textsc{FindSink} starts by assigning $g$ to a maximally subsidized agent. At every round, it computes the pointwise-minimal subsidies for the resulting allocation. If some agent requires a subsidy of two, then $g$ is removed from the previous recipient's bundle and assigned to that agent; otherwise, the procedure stops with the current allocation. The path bound, together with the fact that every original subsidy is at most one, implies that every agent selected in this process is maximally subsidized in the original solution, and hence every tentative allocation considered by \textsc{FindSink} is envy-freeable.

The key point is that no agent can be selected twice. Indeed, if a repetition occurred, the paths associated with the successive selections would yield a zero-weight cycle in $G_A$. Moreover, the argument for the first selection identifies on this cycle an edge entering the initial recipient across which $g$ has marginal value one. Rotating the bundles along this cycle, therefore, preserves welfare and produces an extendable solution, contradicting non-extendability. Hence \textsc{FindSink} terminates within $n$ selections. At termination, the pointwise-minimal subsidies are integral and strictly below two, and thus belong to $\{0,1\}^n$. The procedure is polynomial-time.
\end{proof}

We can now obtain the unit-subsidy guarantee for objective signed dichotomous valuations.

\begin{theorem}
\label{thm:objective-mixed}
For every instance with objective signed dichotomous valuations, there exist a complete allocation $A$ and a subsidy vector $p\in\{0,1\}^n$ such that $(A,p)$ is envy-free and
\[
\sum_{i\in N} p_i\leq n-1.
\]
Moreover, such an allocation and subsidy vector can be computed in polynomial time in the value-oracle model.
\end{theorem}

\begin{proof}
Restrict the instance to the chores $C$. By Theorem~\ref{thm:m-main}, we can compute an envy-free solution $(A^C,p^C)$ with $p^C\in\{0,1\}^n$. Starting from this solution, insert the goods in $G$ one at a time using Lemma~\ref{lem:goods-extension}; after all insertions we obtain a complete envy-free solution $(A,p)$ with $p\in\{0,1\}^n$. If every component of $p$ equals one, subtract one from every subsidy; otherwise some component is already zero. Thus, after normalization, $\sum_i p_i\le n-1$. Polynomial-time computability follows from Theorem~\ref{thm:m-main} and Lemma~\ref{lem:goods-extension}.
\end{proof}

The resulting procedure is stated explicitly as Algorithm~\ref{alg:objective-mixed} in Appendix~\ref{app:objective-mixed-alg}.

The total-subsidy bound in Theorem~\ref{thm:objective-mixed} is tight. Indeed, negative dichotomous chore instances form a subclass of objective signed dichotomous valuations, and the identical binary additive instance with $n$ agents and $n-1$ chores from the introduction already requires a total subsidy of at least $n-1$.

Theorem~\ref{thm:objective-mixed} concerns the subsidy guarantee only. Our construction does not guarantee that the underlying allocation is EF1, since the subsequent goods-allocation procedure of Barman et al.~\cite{BKNS22} does not in general preserve such a guarantee. Pareto optimality also cannot be imposed in general, since Proposition~\ref{prop:po-impossibility} applies to the pure-chore subclass. Under additivity, however, Pareto optimality can be recovered together with the same subsidy bound; see Proposition~\ref{prop:additive-objective-po} in Appendix~\ref{app:additive-objective-po}.

\appendix

\section{Failure of an Incremental Unit-Subsidy Extension}
\label{sec:whynotbkns}

\begin{example}[Failure of a BKNS-type incremental extension]
\label{app:extension-failure}

Consider three agents $N=\{1,2,3\}$ and three chores $M=\{x,y,z\}$. We work in cost form, writing
\( c_i=-v_i \)
for agent \(i\)'s non-negative cost function. Thus, an allocation \(\mathcal A=(A_1,\ldots,A_n)\) together with a subsidy vector \(p=(p_1,\ldots,p_n)\) is envy-free iff
\[ c_i(A_i)-p_i\le c_i(A_j)-p_j
\qquad\text{for all } i,j\in N.
\]
The cost functions depend only on bundle size: \[ c_1(S)=\max\{0,|S|-1\}, \qquad c_2(S)=c_3(S)=\min\{|S|,2\}. \] Thus, \[
\begin{array}{c|cccc}
|S| & 0 & 1 & 2 & 3\\
\hline
c_1(S) & 0 & 0 & 1 & 2\\
c_2(S)=c_3(S) & 0 & 1 & 2 & 2 .
\end{array}
\] Every marginal cost belongs to $\{0,1\}$, so these are dichotomous cost functions. Suppose  $z$ is currently unallocated, and the partial allocation is \[
X_1=\{x,y\},
\qquad
X_2=X_3=\emptyset.
\] Note that this partial allocation is itself reachable from the empty allocation through successive insertions while maintaining the unit-subsidy invariant:
\[
(\emptyset,\emptyset,\emptyset)
\longrightarrow
(\{x\},\emptyset,\emptyset)
\longrightarrow
(\{x,y\},\emptyset,\emptyset),
\] with pointwise-minimal subsidy vectors
\(
(0,0,0), (0,0,0),\text{ and } (1,0,0),
\) respectively. The pointwise-minimal subsidy vector of $(X_1,X_2,X_3)$ is \[
p^*=(1,0,0).
\]

Indeed, agent $1$ has cost $1$ for her own bundle and cost $0$ for either empty bundle, so she requires one unit of subsidy. Agents $2$ and $3$ have no envy, since they have cost $0$ for their own empty bundles and cost $2$ for $X_1$.

We claim that the remaining chore $z$ cannot be inserted so that the resulting allocation admits an envy-free subsidy vector in $\{0,1\}^3$, even if the existing bundles $\{x,y\},\emptyset,\emptyset$ may first be permuted among the agents.

If $z$ is added to the two-chore bundle, the resulting bundle sizes are $(3,0,0)$. Every agent evaluates a three-chore bundle at cost $2$ and an empty bundle at cost $0$. Hence, the agent receiving the three-chore bundle requires a subsidy at least two units larger than that of an agent receiving an empty bundle. Thus, no subsidy vector in $\{0,1\}^3$ can make the resulting allocation envy-free.

It remains to consider the case in which $z$ is added to an empty bundle. The resulting bundle sizes are $(2,1,0)$, up to permutation. If an agent in $\{2,3\}$ receives the two-chore bundle, she evaluates it at cost $2$ and the empty bundle at cost $0$, again requiring a subsidy difference of at least two.

Therefore, the only remaining possibility is that agent $1$ receives the two-chore bundle. Let $i\in\{2,3\}$ receive the singleton and let $j$ receive the empty bundle. Envy-freeness for agent $1$ with respect to $i$ requires
\[
1-p_1\le 0-p_i,
\]
and hence
\[
p_1\ge p_i+1.
\]
Envy-freeness for agent $i$ with respect to $j$ requires \[
1-p_i\le -p_j,
\] and hence
\[
p_i\ge p_j+1.
\] Combining the two inequalities gives
\[
p_1\ge p_j+2,
\] which is impossible for $p\in\{0,1\}^3$. Hence, no assignment of the additional chore $z$, even after an arbitrary permutation of the previously allocated bundles, preserves the unit-subsidy guarantee. Notice that this is an obstruction to the incremental construction, not to the instance itself. The complete allocation
\[
(\{x\},\{y\},\{z\})
\] is envy-free without subsidies: agent $1$ evaluates every singleton at cost $0$, while agents $2$ and $3$ evaluate every singleton at cost $1$. \hfill $\square$
\end{example}

\section{Proof of \texorpdfstring{Lemma~\ref{lem:goods-extension}}{the one-good extension lemma}}
\label{app:goods-extension}
We first restate the lemma for convenience.
\renewcommand{\restatedname}{Lemma~\ref{lem:goods-extension}}
\begin{restated}
Let $M'$ be a set of allocated items and let $A=(A_1,\ldots,A_n)$ be a complete allocation of $M'$. Suppose that the valuations are integer-valued and that there exists $p\in\{0,1\}^n$ such that $(A,p)$ is envy-free. Let $g\notin M'$ satisfy
\[
v_i(S\cup\{g\})-v_i(S)\in\{0,1\}
\]
for every agent $i$ and every $S\subseteq M'$. Then there exist a complete allocation $A'$ of $M'\cup\{g\}$ and $p'\in\{0,1\}^n$ such that $(A',p')$ is envy-free. Moreover, $(A',p')$ can be computed in polynomial time given value-oracle access.
\end{restated}

\begin{proof}
The proof is based on the one-good extension argument of Barman et al.~\cite{BKNS22}. Since that argument is stated there for dichotomous goods instances, we restate its relevant steps in a self-contained way and verify that each of them remains valid when the already allocated bundles may contain chores.

For an allocation $Y=(Y_1,\ldots,Y_n)$, let $G_Y$ denote its weighted envy graph, with edge weights
\[
w_Y(i,j):=v_i(Y_j)-v_i(Y_i).
\]
We use two standard characterizations of envy-freeability
\cite{HS19,BKNS22} (Theorems~\ref{thm:hs-characterisation} and~\ref{thm:hs-minsubsidy}). First, $Y$ is envy-freeable if and only if it maximizes
\[
\sum_{i\in N}v_i(Y_{\sigma(i)})
\]
over all permutations $\sigma$ of its bundles. Equivalently, $G_Y$ contains no positive-weight directed cycle. Second, if $Y$ is envy-freeable, then its pointwise-minimal supporting subsidy for agent $i$ is the maximum weight of a directed path starting at $i$ in $G_Y$.

We may assume without loss of generality that $p$ is the pointwise-minimal subsidy vector supporting $A$. Indeed, starting from the subsidy vector in the statement, replace it by the pointwise-minimal one. This can only decrease the subsidies componentwise. Since the valuations are integer-valued, all edge weights in $G_A$, and therefore all pointwise-minimal subsidies, are integers. Consequently, the resulting vector still belongs to $\{0,1\}^n$.

We first prove two properties of adding $g$ that will be used throughout the proof.

\medskip
\noindent
\textbf{Claim 1.}
Let $Y$ be an envy-freeable allocation of $M'$, and suppose that for some agent $x$,
\[
v_x(Y_x\cup\{g\})-v_x(Y_x)=1.
\]
Then
\[
Z=(Y_1,\ldots,Y_x\cup\{g\},\ldots,Y_n)
\]
is envy-freeable.

\smallskip
\noindent
\emph{Proof of Claim 1.}
Write
\[
\operatorname{SW}(Y):=\sum_{i\in N}v_i(Y_i).
\]
Since the marginal value of $g$ for $x$ on $Y_x$ is one,
\[
\operatorname{SW}(Z)=\operatorname{SW}(Y)+1.
\]
Consider any permutation $\sigma$ of the bundles of $Z$. Let $a$ be the agent who receives the bundle containing $g$ under this permutation. Removing $g$ from that bundle leaves a permutation of the bundles of $Y$. Since $Y$ is envy-freeable, the welfare of this reassignment is at most
$\operatorname{SW}(Y)$. Moreover,
\[
v_a(S\cup\{g\})-v_a(S)\leq 1
\]
for every $S\subseteq M'$. Hence
\[
\operatorname{SW}(Z_\sigma)
\leq \operatorname{SW}(Y)+1
=\operatorname{SW}(Z).
\]
Thus $Z$ maximizes welfare over all permutations of its bundles and is therefore envy-freeable.
\hfill$\square$

\medskip
\noindent
\textbf{Claim 2.}
Let $Y$ be an allocation and let
\[
Z=(Y_1,\ldots,Y_x\cup\{g\},\ldots,Y_n).
\]
Then
\[
w_Z(i,j)=w_Y(i,j)
\qquad
\text{for all }i,j\neq x,
\]
while
\[
w_Z(x,j)\leq w_Y(x,j)
\qquad\text{and}\qquad
w_Z(i,x)\leq w_Y(i,x)+1.
\]
Consequently, for every simple directed path $P$,
\begin{equation}\label{S1}
    w_Z(P)\leq w_Y(P)+1.
\end{equation}

\smallskip
\noindent
\emph{Proof of Claim 2.}
Only the bundle of $x$ changes. Hence, edges not incident to $x$ are unchanged. For every $j\neq x$,
\[
\begin{aligned}
w_Z(x,j)
&=v_x(Y_j)-v_x(Y_x\cup\{g\})\\
&=w_Y(x,j)
-\bigl(v_x(Y_x\cup\{g\})-v_x(Y_x)\bigr)
\leq w_Y(x,j),
\end{aligned}
\]
because the marginal value of $g$ is nonnegative. Similarly, for
$i\neq x$,
\[
\begin{aligned}
w_Z(i,x)
&=v_i(Y_x\cup\{g\})-v_i(Y_i)\\
&=w_Y(i,x)
+\bigl(v_i(Y_x\cup\{g\})-v_i(Y_x)\bigr)
\leq w_Y(i,x)+1.
\end{aligned}
\]
A simple path enters $x$ at most once. Its weight can therefore increase by at most one, proving \eqref{S1}.
\hfill$\square$

\medskip
We now follow the extendable/non-extendable distinction of Barman et al.~\cite{BKNS22}. For a subsidy vector $q$, write
\[
M(q):=\{i\in N:q_i\geq q_j\text{ for every }j\in N\}.
\]
Call $(A,p)$ \emph{extendable with $g$} if there exists a permutation $\sigma$ such that, writing
\[
B_i:=A_{\sigma(i)}
\qquad\text{and}\qquad
q_i:=p_{\sigma(i)},
\]
the pair $(B,q)$ is envy-free and there is some
$\kappa\in M(q)$ satisfying
\begin{equation}\label{S2}
 v_\kappa(B_\kappa\cup\{g\})-v_\kappa(B_\kappa)=1.
\end{equation}

We will also use the following elementary permutation property. If $(A,p)$ is envy-free and $B=A_\sigma$ is envy-freeable, then $(B,p_\sigma)$ is envy-free. To see this, the envy-freeness of
$(A,p)$ gives
\[
v_i(A_i)+p_i
\geq
v_i(A_{\sigma(i)})+p_{\sigma(i)}
\qquad\text{for every }i.
\]
Since both $A$ and $B$ are envy-freeable, each maximizes
welfare among permutations of the same family of bundles.
Consequently,
\[
\sum_{i\in N}v_i(A_i)
=
\sum_{i\in N}v_i(B_i).
\]
Summing the preceding inequalities, the subsidies cancel, and the equality of welfare implies that every inequality must hold with equality. The envy-freeness inequalities for $(B,p_\sigma)$ then follow directly. See Lemma 3 in Barman et al.~\cite{BKNS22} for a detailed argument.

\medskip
\noindent
\textbf{Case 1: $(A,p)$ is extendable with $g$.}

Let $\sigma$ and $\kappa$ satisfy the above conditions, and define
\[
C=(B_1,\ldots,B_\kappa\cup\{g\},\ldots,B_n).
\]
By Claim~1, $C$ is envy-freeable.

It remains to show that $C$ can be supported by subsidies in $\{0,1\}^n$. Since $q_\kappa$ is maximal and
$q\in\{0,1\}^n$, there are two cases.

If $q_\kappa=0$, then $q_i=0$ for all $i\in N$. Thus, every directed path in $G_B$ has weight at most zero. By Claim~2, every directed path in $G_C$ has weight at most one. Hence, the pointwise-minimal supporting subsidies for $C$ are at most one. They are nonnegative integers, so they belong to $\{0,1\}^n$.

Suppose next that $q_\kappa=1$. Define
\[
\widehat q_\kappa:=0,
\qquad
\widehat q_i:=q_i
\quad (i\neq\kappa).
\]
We claim that $(C,\widehat q)$ is envy-free. For $i,j\neq\kappa$, nothing changes from $(B,q)$. For agent
$\kappa$, using \eqref{S2},
\[
v_\kappa(C_\kappa)+\widehat q_\kappa
=
v_\kappa(B_\kappa)+1
=
v_\kappa(B_\kappa)+q_\kappa
\geq
v_\kappa(B_j)+q_j.
\]
Finally, for every $i\neq\kappa$,
\[
\begin{aligned}
v_i(C_i)+\widehat q_i
&=v_i(B_i)+q_i\\
&\geq v_i(B_\kappa)+q_\kappa\\
&=v_i(B_\kappa)+1\\
&\geq v_i(B_\kappa\cup\{g\})
 =v_i(C_\kappa)+\widehat q_\kappa,
\end{aligned}
\]
where the last inequality follows because the marginal value of $g$ is at most one. Thus $\widehat q\in\{0,1\}^n$ supports $C$.

Hence, the desired extension exists in the extendable case.

\medskip
\noindent
\textbf{Case 2: $(A,p)$ is not extendable with $g$.}

We use the \textsc{FindSink} procedure of Barman et al.~\cite{BKNS22}, and verify that its proof remains valid in the present setting.

For $s\in N$, write
\[
X^s:=(A_1,\ldots,A_s\cup\{g\},\ldots,A_n).
\]
Start with an arbitrary $s_0\in M(p)$. Whenever $X^{s_\tau}$ is envy-freeable, let $\phi^\tau$ denote its pointwise-minimal supporting subsidy vector. If every component of $\phi^\tau$ is strictly smaller than $2$, terminate. Otherwise, choose an agent
$s_{\tau+1}$ with
\[
\phi^\tau_{s_{\tau+1}}\geq2
\]
and continue with $X^{s_{\tau+1}}$.

We first establish the following:
\begin{equation}\label{S3}
    s_\tau\in M(p)
\quad\text{and}\quad
X^{s_\tau}\text{ is envy-freeable}
\end{equation}
for every agent considered by the procedure.

For $\tau=0$, the first part follows from the choice of $s_0$. Suppose, towards a contradiction, that $X^{s_0}$ is not envy-freeable. By the welfare characterization, there exists a permutation of its bundles whose welfare is strictly larger. Since all valuations are integer-valued, the improvement is at least one.

Let $B$ be the corresponding permutation of the bundles of $A$, obtained by deleting $g$ from the permuted bundle containing it, and let $k$ be the agent who receives the bundle $A_{s_0}$ in $B$. Since the marginal value of $g$ is nonnegative,
\[
\operatorname{SW}(X^{s_0})\geq\operatorname{SW}(A).
\]
Therefore,
\begin{equation}\label{S4}
\operatorname{SW}(B)
+
\bigl(v_k(B_k\cup\{g\})-v_k(B_k)\bigr)
\geq
\operatorname{SW}(A)+1.
\end{equation}
On the other hand, $A$ is envy-freeable, and hence
\[
\operatorname{SW}(B)\leq\operatorname{SW}(A).
\]
The marginal term in \eqref{S4} is at most one. Consequently, both inequalities must be tight:
\[
\operatorname{SW}(B)=\operatorname{SW}(A)
\]
and
\[
v_k(B_k\cup\{g\})-v_k(B_k)=1.
\]
Thus $B$ is envy-freeable. By the permutation property above, if $q$ is obtained by permuting $p$ in the same way, then $(B,q)$ is envy-free. Moreover, $B_k=A_{s_0}$ and $s_0\in M(p)$, so $k\in M(q)$. Hence $(A,p)$ is extendable with $g$, a contradiction. Therefore $X^{s_0}$ is envy-freeable.

Now suppose \eqref{S3} holds for $s_\tau$ and the procedure selects $s_{\tau+1}$ because $\phi^\tau_{s_{\tau+1}}\geq2$. We claim that $s_{\tau+1}\in M(p)$. If not, then the maximal subsidy in $p$ is one and $p_{s_{\tau+1}}=0$. Since $p$ is pointwise-minimal, every path starting from $s_{\tau+1}$ in $G_A$ has weight at most zero. By Claim~2, every path starting from $s_{\tau+1}$ in $G_{X^{s_\tau}}$ has weight at most one. Hence $\phi^\tau_{s_{\tau+1}}\leq1$, a contradiction. Thus $s_{\tau+1}\in M(p)$. The same argument used for $s_0$ then shows that $X^{s_{\tau+1}}$ is envy-freeable. This proves \eqref{S3}.

It remains to show that the procedure terminates. Suppose otherwise that some agent is selected twice. Consider two consecutive occurrences of a repeated agent and relabel the corresponding segment
as
\[
s_0,s_1,\ldots,s_T,
\qquad
s_T=s_0,
\]
with $s_0,\ldots,s_{T-1}$ distinct.

For every $\tau=1,\ldots,T$, agent $s_\tau$ is selected because its pointwise-minimal subsidy for $X^{s_{\tau-1}}$ is at least two. Hence, there exists a simple directed path $Q_\tau$ starting
at $s_\tau$ in $G_{X^{s_{\tau-1}}}$ such that
\begin{equation}\label{S6}
w_{X^{s_{\tau-1}}}(Q_\tau)\geq2.
\end{equation}
The path $Q_\tau$ must contain $s_{\tau-1}$. Otherwise, none of its edge weights differs from the corresponding edge weights in $G_A$, whereas every path starting at $s_\tau$ in $G_A$ has a weight at most 1, contradicting \eqref{S6}.

Let $P_\tau$ be the initial segment of $Q_\tau$ from $s_\tau$ to $s_{\tau-1}$. By Claim~2,
\begin{equation}\label{S7}
 w_A(Q_\tau)
\geq
w_{X^{s_{\tau-1}}}(Q_\tau)-1
\geq1.
\end{equation}
The part of $Q_\tau$ following $P_\tau$ is a path starting at $s_{\tau-1}$ in $G_A$, and therefore has weight at most $p_{s_{\tau-1}}\leq1$. It follows from \eqref{S7} that
\begin{equation}\label{S8}
w_A(P_\tau)\geq0.
\end{equation}

Now concatenate the paths $P_T,P_{T-1},\ldots,P_1$. Since $s_T=s_0$, this concatenation starts and ends at $s_0$. If some vertex is repeated in the concatenation, split the resulting closed sequence of edges at repeated vertices. In this way, it can be decomposed into directed cycles of $G_A$. By \eqref{S8}, the total weight of the concatenated sequence is nonnegative. On the other hand, since $A$ is envy-freeable, $G_A$ contains no positive-weight directed cycle. Therefore, every cycle arising in the above decomposition has weight at most zero, while their total weight is nonnegative. It follows that the total weight is zero, and consequently, every cycle in the decomposition has weight zero.

We now derive a contradiction to non-extendability. Since $s_1$ was selected after considering $X^{s_0}$, there is a path $Q_1$ satisfying \eqref{S6}. Its weight in $G_A$ is at most $p_{s_1}\leq1$, while Claim~2 says that its weight can increase by at most one after $g$ is added to $A_{s_0}$. Hence
\[
w_{X^{s_0}}(Q_1)=2,
\qquad
w_A(Q_1)=1,
\]
and the weight of $Q_1$ increases by exactly one. Only an edge entering $s_0$ can increase when $g$ is added to $A_{s_0}$; edges leaving $s_0$ can only decrease. Therefore, if $k$ is the predecessor of $s_0$ on $Q_1$, then
\begin{equation}\label{S9}
v_k(A_{s_0}\cup\{g\})-v_k(A_{s_0})=1.
\end{equation}
Since the edge $(k,s_0)$ belongs to $P_1$, it belongs to one of the zero-weight directed cycles obtained above. Let $C$ be such a cycle. Cyclically permuting the bundles of $A$ along $C$ preserves social welfare, because the change in welfare is exactly $w_A(C)=0$. Let $\sigma$ denote the resulting permutation and let $B=A_\sigma$. Hence
\[
\operatorname{SW}(B)=\operatorname{SW}(A),
\] so $B$ is envy-freeable. By the permutation property, $(B,p_\sigma)$ is envy-free.

Furthermore, the cycle contains $(k,s_0)$, so agent $k$ receives bundle $A_{s_0}$ in $B$. Thus \eqref{S9} gives
\[
v_k(B_k\cup\{g\})-v_k(B_k)=1.
\]
Since $s_0\in M(p)$, agent $k$ belongs to $M(p_\sigma)$. Consequently, $\sigma$ and $k$ witness that $(A,p)$ is extendable with $g$, contradicting the assumption of Case~2.

Therefore, no agent is selected twice, and the procedure terminates after at most $n$ selections. At termination, the current allocation $X^s$ is envy-freeable by \eqref{S3}, and every component of its pointwise-minimal subsidy vector is strictly smaller than $2$. These subsidies are nonnegative integers, and hence they belong to $\{0,1\}^n$. This proves the existence of the desired allocation $A'$ and subsidy vector $p'$ in the non-extendable case.

Finally, all steps are polynomial-time. Pointwise-minimal subsidies can be computed from maximum-weight paths in an envy-freeable allocation. Extendability can be tested by the \textsc{Extend} subroutine of Barman et al.~\cite{BKNS22}: for each candidate recipient and most-subsidized bundle, it solves a maximum-weight bipartite matching problem. Its correctness uses only the welfare characterization of envy-freeability and the values of the existing bundles, and is therefore unchanged here. In the non-extendable case, the procedure above considers at most $n$ agents. Thus, the entire one-good extension can be carried out in polynomial time using value-oracle access.
\end{proof}

\section{Algorithm for Objective Signed Dichotomous Valuations}
\label{app:objective-mixed-alg}
We first restate the theorem for convenience.
\renewcommand{\restatedname}{Theorem~\ref{thm:objective-mixed}}
\begin{restated}
For every instance with objective signed dichotomous valuations, there exist a complete allocation $\alloc$ and a subsidy vector $p\in\{0,1\}^n$ such that $(\alloc,p)$ is envy-free and
\[
\sum_{i\in N} p_i\leq n-1.
\]
Moreover, such an allocation and subsidy vector can be computed in
polynomial time in the value-oracle model.
\end{restated}
We next give the complete algorithm underlying Theorem~\ref{thm:objective-mixed}, using Lemma~\ref{lem:goods-extension} as the one-good extension subroutine.

For brevity, let \textsc{ChoreAllocate} denote the polynomial-time procedure developed in Section~\ref{sec:main} (Algorithm~\ref{alg:main}) for negative dichotomous valuations. Given a negative dichotomous chore instance, it returns a complete allocation $\alloc$ and $p\in\{0,1\}^n$ such that $(\alloc,p)$ is envy-free and $\sum_i p_i\le n-1$.

\begin{algorithm}[H]
\caption{Unit subsidies for objective signed dichotomous valuations}
\label{alg:objective-mixed}
\begin{algorithmic}[1]
\Require Agents $N=[n]$, items $M=G\mathbin{\dot\cup}C$, and value-oracle access to objective signed dichotomous valuations $\{v_i\}_{i\in N}$.
\Ensure A complete allocation $\alloc$ and $p\in\{0,1\}^n$ such that $(\alloc,p)$ is envy-free and $\sum_{i\in N}p_i\le n-1$.

\State Restrict each $v_i$ to subsets of $C$.
\State $(\alloc,p)\gets\textsc{ChoreAllocate}(N,C,\{v_i|_C\}_{i\in N})$.
\State Replace $p$ by the pointwise-minimal subsidy vector supporting $\alloc$.

\For{each $g\in G$}
    \State Apply the one-good extension procedure of Lemma~\ref{lem:goods-extension} to $(\alloc,p)$ and $g$, obtaining an allocation $\alloc'$ of the currently allocated items together with $g$ and a subsidy vector $p'\in\{0,1\}^n$ such that $(\alloc',p')$ is envy-free.
    \State $\alloc\gets\alloc'$.
    \State Replace $p'$ by the pointwise-minimal subsidy vector supporting $\alloc$, and set $p\gets p'$.
\EndFor

\State \Return $(\alloc,p)$.
\end{algorithmic}
\end{algorithm}

The one-good extension in Line~5 is implemented by the \textsc{Extend}/\textsc{FindSink} procedure of Barman et al.~\cite{BKNS22}, whose validity in the present setting is established in the proof of Lemma~\ref{lem:goods-extension} (Appendix~\ref{app:goods-extension}). After every iteration, we take the pointwise-minimal supporting subsidy vector. Since the valuations are integer-valued, this vector is integral, and Lemma~\ref{lem:goods-extension} implies that all its components are at most one. Moreover, pointwise minimality implies that at least one component is zero. Hence, throughout the algorithm
\[
p\in\{0,1\}^n
\qquad\text{and}\qquad
\sum_{i\in N}p_i\le n-1.
\]

\section{Pareto Optimality and Unit Subsidies for Additive Objective Signed Dichotomous Valuations}
\label{app:additive-objective-po}
The unit-subsidy guarantee for additive mixed instances is already known from Lu et al.~\cite{LMS26}. For the objective signed dichotomous subclass, we provide a direct construction showing that the underlying allocation can additionally be chosen Pareto optimal.
\begin{proposition}
\label{prop:additive-objective-po}
For every instance with additive objective signed dichotomous valuations, there exist a complete allocation $\alloc$ and a subsidy vector $p\in\{0,1\}^n$ such that $(\alloc,p)$ is envy-free,
\[
\sum_{i\in N}p_i\leq n-1,
\]
and $\alloc$ is Pareto optimal. Moreover, such an outcome can be computed in polynomial time.
\end{proposition}

\begin{proof}
The proof has three steps. First, allocate the goods using Halpern and Shah's binary-additive result, obtaining an envy-freeable allocation with unit subsidies. Next, assign every chore that is costless to some agent to such an agent. Finally, assign each universally costly chore to an agent currently receiving zero subsidy and raise that agent's subsidy from $0$ to $1$; if all subsidies become $1$, subtract one from every subsidy. This preserves envy-freeness and keeps total subsidy at most $n-1$. Since every item is assigned to an agent attaining its maximum possible value, the final allocation is Pareto optimal.

We now formally write the details. Let $M=G\mathbin{\dot\cup}C$ be the common partition into goods and chores. Since the valuations are additive, for every agent $i$,
\[
v_i(g)\in\{0,1\}\qquad (g\in G)
\]
and
\[
v_i(c)\in\{0,-1\}\qquad (c\in C).
\]

We first allocate the goods. Ignore, for the moment, goods that have value zero for every agent, and apply the binary-additive goods result of Halpern and Shah~\cite{HS19} to the remaining goods. Their construction returns a non-wasteful allocation $\alloc^G$ that is envy-freeable and whose minimal subsidies have total at most $n-1$. Moreover, their proof shows that every path in the corresponding weighted envy graph has weight at most one. Since all edge weights are integral, the pointwise-minimal subsidies are integral as well, and hence we may choose
\[
p\in\{0,1\}^n,
\qquad
\sum_{i\in N}p_i\leq n-1.
\]
Goods having value zero for every agent may now be allocated arbitrarily, without affecting either envy-freeness or the subsidy vector.

We next allocate the chores. Partition $C$ into
\[
C^0:=\{c\in C:\text{there exists }i\in N\text{ with }v_i(c)=0\}
\]
and
\[
C^-:=\{c\in C:v_i(c)=-1\text{ for every }i\in N\}.
\]
Thus, $C^-$ is the set of universally costly chores.

First, consider a chore $c\in C^0$. Choose an agent $k$ with $v_k(c)=0$ and assign $c$ to $k$. This preserves envy-freeness with the same subsidy vector. Indeed, agent $k$'s value for her own bundle is unchanged, whereas for every agent $i\neq k$,
\[
v_i(A_k\cup\{c\})
=
v_i(A_k)+v_i(c)
\leq v_i(A_k).
\]
Hence, $k$'s bundle becomes weakly less desirable to every other agent, while all other relevant bundle values remain unchanged. Repeating this operation allocates all chores in $C^0$ while preserving the envy-free solution $(\alloc,p)$.

It remains to allocate the universally costly chores in $C^-$. We maintain the following invariant:
\begin{equation}\label{eq:additive-mixed-invariant}
(\alloc,p)\text{ is envy-free},\qquad
p\in\{0,1\}^n,
\qquad
\text{and }p_k=0\text{ for some }k\in N.
\end{equation}
The invariant holds initially. In particular, the pointwise-minimal subsidy vector has at least one zero component.

Let $c\in C^-$ be an unallocated chore, and choose an agent $k$ with
$p_k=0$. Assign $c$ to $k$ and define
\[
\widehat p_k:=1,
\qquad
\widehat p_i:=p_i\quad(i\neq k).
\]
We claim that the resulting solution remains envy-free.

Since $c$ has value $-1$ for every agent,
\[
v_k(A_k\cup\{c\})+\widehat p_k
=
v_k(A_k)-1+1
=
v_k(A_k)+p_k.
\]
Thus, agent $k$'s subsidized value for her own bundle is unchanged. Moreover, for every $i\neq k$,
\[
v_i(A_k\cup\{c\})+\widehat p_k
=
v_i(A_k)-1+1
=
v_i(A_k)+p_k.
\]
Hence, the subsidized value of $k$'s bundle is unchanged from the viewpoint of every agent. All other bundles and subsidies are unchanged. Consequently, every envy-freeness inequality is exactly the same as before, and the new solution is envy-free.

If $\widehat p$ has a zero component, set $p:=\widehat p$ and continue. Otherwise, if $\widehat p_i=1$ for all $i\in N$, replace it with
\[
p_i:=\widehat p_i-1=0 \text{ for all } i\in N.
\]
Subtracting the same amount from every agent's subsidy leaves all envy-freeness inequalities unchanged. Thus, the invariant \eqref{eq:additive-mixed-invariant} is restored. Repeating this procedure allocates every chore in $C^-$ and produces a complete allocation $\alloc$ with
\[
p\in\{0,1\}^n
\qquad\text{and}\qquad
\sum_{i\in N}p_i\leq n-1.
\]

It remains to prove Pareto optimality. By construction, every good that is valued at one by some agent is assigned to an agent who values it at one; a good valued at zero by every agent may be assigned arbitrarily. Every chore in $C^0$ is assigned to an agent who values it at zero, while every chore in $C^-$ has value $-1$ for every agent regardless of its recipient. Therefore, every item is assigned to an agent attaining the maximum possible value for that item. By additivity,
\[
\sum_{i\in N}v_i(A_i)
=
\sum_{e\in M}\max_{i\in N}v_i(e),
\]
which is the maximum possible utilitarian welfare over all complete allocations. Hence $\alloc$ is Pareto optimal.

Finally, the binary-additive goods allocation and its supporting subsidies can be computed in polynomial time by the construction of Halpern and Shah~\cite{HS19}. The subsequent chore-allocation steps require only identifying a zero-valued recipient for each chore and updating a binary subsidy vector. Hence, the entire construction runs in polynomial time.
\end{proof}

\bibliographystyle{alpha}
\bibliography{references}

\end{document}